\documentclass[a4paper,UKenglish]{lipics-v2016}
 
\usepackage{microtype}

\usepackage{tikz} 
\usepackage{amsmath}
\usepackage{amssymb}
\usepackage{babel}
\usepackage{refcount}
\usetikzlibrary{arrows, automata, positioning}

\newcommand{\va}{\boldsymbol{a}}
\newcommand{\vr}{\boldsymbol{r}}
\newcommand{\vx}{\boldsymbol{x}}
\renewcommand{\Im}{\mathrm{Im}}
\renewcommand{\Re}{\mathrm{Re}}

\newtheorem{conj}[theorem]{Conjecture}

\newcommand{\rats}{\mathbb{Q}}

\newcommand{\defn}{:=}

\newtheorem{proposition}[theorem]{Proposition}

\title{On the Skolem Problem for Continuous Linear Dynamical Systems}

\titlerunning{On the Skolem Problem for Continuous Linear Dynamical Systems} 

\author[1]{Ventsislav Chonev}
\author[2]{Jo\"{e}l Ouaknine}
\author[2]{James Worrell}
\affil[1]{Institute of Science and Technology, Austria}
\affil[2]{University of Oxford, UK}
\authorrunning{V. Chonev and J. Ouaknine and J. Worrell}

\Copyright{Ventsislav Chonev and Jo\"{e}l Ouaknine and James Worrell}


\EventEditors{Ioannis Chatzigiannakis, Michael Mitzenmacher, Yuval
Rabani, and Davide Sangiorgi}
\EventNoEds{4}
\EventLongTitle{43rd International Colloquium on Automata, Languages,
and Programming (ICALP 2016)}
\EventShortTitle{ICALP 2016}
\EventAcronym{ICALP}
\EventYear{2016}
\EventDate{July 11--15, 2016}
\EventLocation{Rome, Italy}
\EventLogo{}
\SeriesVolume{55}
\ArticleNo{XXX}

\begin{document}

\maketitle

\begin{abstract}
  The Continuous Skolem Problem asks whether a real-valued function
  satisfying a linear differential equation has a zero in a given
  interval of real numbers.  This is a fundamental reachability
  problem for continuous linear dynamical systems, such as linear
  hybrid automata and continuous-time Markov chains.  Decidability of
  the problem is currently open---indeed decidability is open even for
  the sub-problem in which a zero is sought in a bounded interval.  In
  this paper we show decidability of the bounded problem subject to
  Schanuel's Conjecture, a unifying conjecture in transcendental
  number theory.  We furthermore analyse the unbounded problem in
  terms of the frequencies of the differential equation, that is, the
  imaginary parts of the characteristic roots.  We show that the
  unbounded problem can be reduced to the bounded problem if there is
  at most one rationally linearly independent frequency, or if
  there are two rationally linearly independent frequencies and all
  characteristic roots are simple.  We complete the picture by showing
  that decidability of the unbounded problem in the case of two (or
  more) rationally linearly independent frequencies would entail a
  major new effectiveness result in Diophantine approximation, namely
  computability of the Diophantine-approximation types of all real
  algebraic numbers.
\end{abstract}

\section{Introduction}
The Continuous Skolem Problem is a fundamental decision problem
concerning reachability in continuous-time linear dynamical systems.
The problem asks whether a real-valued function satisfying an ordinary
linear differential equation has a zero in a given interval of real
numbers.  More precisely, an instance of the problem comprises an
interval $I\subseteq \mathbb{R}_{\geq 0}$ with rational endpoints,
an ordinary differential equation
\begin{gather} 
f^{(n)} + a_{n-1} f^{(n-1)} + \ldots + a_0 f = 0 
\label{eq:ode}
\end{gather}
whose coefficients are real algebraic, together with initial
conditions $f(0),\ldots,f^{(n-1)}(0)$ that are also real algebraic
numbers.  Writing $f:\mathbb{R}_{\geq 0} \rightarrow \mathbb{R}$ for
the unique solution of the differential equation subject to the
initial conditions, the question is whether there exists $t\in I$ such
that $f(t)=0$.  Decidability of this problem is currently open.
Decidability of the sub-problem in which the interval $I$ is bounded,
called the Bounded Continuous Skolem Problem, is also open~\cite[Open
Problem 17]{pisot}.

The nomenclature \emph{Continuous Skolem Problem} is based on an
analogy with the Skolem Problem for linear recurrence sequences, which
asks whether a given linear recurrence sequence has a zero
term~\cite{TUCS05}.  Whether the latter problem is decidable is an
outstanding question in number theory and theoretical computer
science; see, e.g., the exposition of Tao~\cite[Section 3.9]{Tao08}.

The continuous dynamics of linear hybrid automata and the evolution of
continuous-time Markov chains, amongst many other examples, are
determined by linear differential equations of the form
$\boldsymbol{x}'(t)=A\boldsymbol{x}(t)$, where
$\boldsymbol{x}(t) \in \mathbb{R}^n$ and $A$ is an $n\times n$ matrix
of real numbers~\cite{Alur15}.  A basic reachability question in this
context is whether, starting from an initial state
$\boldsymbol{x}(0)$, the system reaches a given hyperplane
$\{ \boldsymbol{y} \in \mathbb{R}^n : \boldsymbol{u}^T\boldsymbol{y} =
0 \}$ with normal vector $\boldsymbol{u}\in\mathbb{R}^n$.  For
example, one can ask whether the continuous flow of a hybrid automaton
leads to a particular transition guard being
satisfied or an invariant being violated.  Now the function
$f(t) = \boldsymbol{u}^T \boldsymbol{x}(t)$ satisfies a linear
differential equation of the form (\ref{eq:ode}), and it turns out
that the hyperplane reachability problem is inter-reducible with the
Continuous Skolem Problem (see~\cite[Theorem 6]{pisot} for further
details).  Moreover, under this reduction the Bounded Continuous
Skolem Problem corresponds to a time-bounded version of the hyperplane
reachability problem.

The \emph{characteristic polynomial} of the differential equation
(\ref{eq:ode}) is
\[\chi(x) := x^n + a_{n-1} x^{n-1} + \ldots + a_0 \, .\]
Let $\lambda_1,\ldots,\lambda_m$ be the distinct roots of $\chi$.  Any
solution of (\ref{eq:ode}) has the form $f(t)=\sum_{j=1}^m P_j(t)
e^{\lambda_j t}$, where the $P_j$ are polynomials with algebraic
coefficients that are determined by the initial conditions of the
differential equation.  We call a function $f$ in this form an
\emph{exponential polynomial}.  If the roots of $\chi$ are all simple
then $f$ can be written as an exponential polynomial in which
the polynomials $P_j$ are all constant.

The Continuous Skolem Problem can equivalently be formulated in terms
of whether an exponential polynomial has a zero in a given interval of
reals.  If the characteristic roots have the
form $\lambda_j=r_j+i\omega_j$, where $r_j,\omega_j\in\mathbb{R}$,
then we can also write $f(t)$ in the form
$f(t)=\sum_{j=1}^m e^{r_jt}(Q_{1,j}(t)\sin(\omega_j t) +
Q_{2,j}(t)\cos(\omega_jt))$, where the polynomials $Q_{1,j},Q_{2,j}$
have real algebraic coefficients.  We call $\omega_1,\ldots,\omega_m$
the \emph{frequencies} of $f$.

Our first result is to show decidability of the Bounded Continuous
Skolem Problem subject to Schanuel's Conjecture, a unifying conjecture
in transcendental number theory that plays a key role in the study of
the exponential function on both real and complex
numbers~\cite{Zilber02,Zilber05}.  Intuitively, decidability of the
Bounded Continuous Skolem Problem is non-trivial because an
exponential polynomial can approach $0$ tangentially. 
Assuming Schanuel's Conjecture, we show that any exponential
polynomial admits a factorisation such that the zeros of each factor
can be detected using finite-precision numerical computations.  Our
method, however, does not enable us to bound the precision required to
find zeros, so we do not obtain a complexity bound for the procedure.

A celebrated paper of Macintyre and Wilkie~\cite{macintyreWilkie}
obtains decidability of the first-order theory of
$\mathbb{R}_{\mathrm{exp}}=(\mathbb{R},0,1,<,\,\cdot\,,+,\exp)$
assuming Schanuel's Conjecture over $\mathbb{R}$.  The proof
of~\cite[Theorem 3.1]{TuringSchanuel} mentions an unpublished result
of Macintyre and Wilkie that generalises~\cite{macintyreWilkie} to
obtain decidability when $\mathbb{R}_{\mathrm{exp}}$ is augmented with
the restricted functions ${\sin}\!\restriction_{[0,2\pi]}$ and
${\cos}\!\restriction_{[0,2\pi]}$, this time assuming Schanuel's
Conjecture over $\mathbb{C}$.  This result immediately implies
(conditional) decidability of the Bounded Continuous Skolem Problem.
However, decidability of latter problem is simpler and, as we show
below, can be established more directly.

In the unbounded case we analyse exponential polynomials in terms of
the number of rationally linearly independent frequencies.  We show
that the unbounded problem can be reduced to the bounded problem if
there is at most one rationally linearly independent frequency, or if
there are two rationally linearly independent frequencies and all
characteristic roots are simple.  These two reductions are
unconditional and rely on the cell decomposition theorem for
semi-algebraic sets~\cite{BasuPR06} and Baker's Theorem on linear
forms in logarithms of algebraic numbers~\cite{baker}.

We complete the picture by showing that decidability of the unbounded
problem in the case of two (or more) rationally linearly independent
frequencies would entail a major new effectiveness result in
Diophantine approximation---namely computability of the
Diophantine-approximation types of all real algebraic numbers.  As we
discuss in Appendix~\ref{sec:hardness}, currently essentially nothing is known
about Diophantine-approximation types of algebraic numbers of degree
three or higher, and they are the subject of several longstanding open
problems.

The question of deciding whether an exponential polynomial $f$ has
infinitely many zeros is investigated in~\cite{ChonevOW16}.  There the
problem is shown to be decidable if $f$ satisfies a differential
equation of order at most $7$.  This result does not rely on
Schanuel's Conjecture.  It is also shown in~\cite{ChonevOW16} that,
analogously with the Continuous Skolem Problem, decidability of the
Infinite Zeros Problem in the general case would entail significant
new effectiveness results in Diophantine approximation.

\section{Mathematical Background}
\label{sec:Bounded}
\subsection{Zero Finding}
\label{sec:zero-finding}
Let $f:[a,b]\rightarrow \mathbb{R}$ be a function defined on a closed
interval of reals with endpoints $a,b\in\mathbb{Q}$.  Suppose the
following two conditions hold: (i)~there exists $M>0$ such that $f$ is
$M$-Lipschitz, i.e., $|f(s)-f(t)|\leq M|s-t|$ for all $s,t \in [a,b]$;
(ii)~given $t\in [a,b]\cap \mathbb{Q}$ and positive error bound
$\varepsilon\in \mathbb{Q}$, we can compute $q\in\mathbb{Q}$ such that
$|f(t)-q|<\varepsilon$.  Then
given a positive rational number $\delta$ we can compute
piecewise linear functions $f_\delta^+,f_\delta^{-} : [a,b] \rightarrow \mathbb{R}$
such that $f_\delta^{-}(t) \leq f(t) \leq f_\delta^{+}(t)$ and $f_\delta^+(t)-f_\delta^{-}(t)\leq \delta$
for all $t \in [a,b]$.  We do this as follows:
\begin{enumerate}
\item
Pick $N\in\mathbb{N}$ such that $\frac{1}{N} < \frac{\delta}{4(b-a)M}$ and
consider sample points
$s_j := a + \frac{(b-a)j}{N}$, $j=0,\ldots,N$, dividing the
interval $[a,b]$ into $N$ sub-intervals, each of length at most
$\frac{\delta}{4M}$.  
\item For each sample point $s_j$ compute $q_j\in\mathbb{Q}$
  such that $|q_j - f(s_j)| < \frac{\delta}{4}$, define
  $f_\delta^{-}(s_j) = q_j - \frac{\delta}{2}$,
  $f_\delta^{+}(s_j) = q_j + \frac{\delta}{2}$, and extend $f_\delta^{-}$ and
  $f_\delta^{+}$ linearly between sample points.
\end{enumerate}
Note that the Lipschitz condition on $f$ ensures that
$f_\delta^{-} \leq f \leq f_\delta^{+}$.

Now suppose that $f$ satisfies the following additional conditions:
(iii)~$f(a)\neq 0,f(b) \neq 0$; (iv)~for any $t\in(a,b)$ such that $f(t)=0$,
$f'(t)$ exists and is non-zero, i.e., $f$ has no tangential zeros.
Then we can decide the existence of a zero of $f$ by computing upper
and lower approximations $f_\delta^{+}$ and $f_\delta^{-}$ for successively smaller
values of $\delta$.  If $f_\delta^{+}(t)<0$ for all $t$ or $f_\delta^{-}(t)>0$
for all $t$ then we conclude that $f$ has no zero on $[a,b]$; if
$f_\delta^{+}(s)<0$ and $f_\delta^{-}(t)>0$ for some $s,t$ then we conclude that $f$
has a zero; otherwise we proceed to a smaller value of $\delta$.
This procedure terminates since by (iii) and (iv) either $f$ has a
zero in $[a,b]$ or it is bounded away from zero.

\subsection{Number-Theoretic Algorithms}
\label{sec:number-theory}
For the purposes of establishing decidability, we can assume that an
instance of the Continuous Skolem Problem is a real-valued exponential
polynomial $f(t)=\sum_{j=1}^m P_j(t) e^{\lambda_j t}$, where
$\lambda_1,\ldots,\lambda_m$ and the coefficients of the polynomials
$P_1,\ldots,P_m$ are algebraic, see~\cite[Theorem 6]{pisot}.

For computational purposes we represent an algebraic number $\alpha$
by a polynomial $P$ with rational coefficients such that
$P(\alpha)=0$, together with a numerical approximation $p+qi$, where
$p,q\in \mathbb{Q}$, of sufficient accuracy to distinguish $\alpha$
from the other roots of $P$~\cite[Section 4.2.1]{Coh93}.  Given this
representation we can obtain approximations of $\alpha$ to within an
arbitrarily small additive error.

Let $K$ be the extension field of $\mathbb{Q}$ generated by
$\lambda_1,\ldots,\lambda_m$ and the coefficients of the polynomials
$P_1,\ldots,P_m$.  Note that $K$ is closed under complex conjugation.
We can compute a primitive element of $K$, that is,
an algebraic number $\theta$ such that $K=\mathbb{Q}(\theta)$,
together with a representation of each characteristic root $\lambda_j$
as a polynomial in $\theta$ with rational coefficients
(see~\cite[Section 4.5]{Coh93}).  From the representation
of $\lambda_1,\ldots,\lambda_m$ as elements of $\mathbb{Q}(\theta)$,
it is straightforward to determine maximal $\mathbb{Q}$-linearly
independent subsets of $\{ \Re(\lambda_j) : 1 \leq j \leq m \}$ and
$\{ \Im(\lambda_j) : 1 \leq j \leq m \}$ (see~\cite[Section
1]{Just89}).

Let $\log$ denote the branch of the complex logarithm defined by
$\log(re^{i\theta}) = \log(r) + i\theta$ for a positive real number
$r$ and $0 \leq \theta < 2\pi$.  Recall that one can compute $\log z$
and $e^{z}$ to within arbitrarily small additive error given a
sufficiently precise approximation of $z$~\cite{Brent76}.

\subsection{Laurent Polynomials}
\label{sec:laurent}
Let $K$ be a sub-field of $\mathbb{C}$ that has finite dimension over
$\mathbb{Q}$ and is closed under complex conjugation.  Fix
non-negative integers $r$ and $s$, and consider a single variable $x$
and tuples of variables
$\boldsymbol{y}= \langle y_1,\ldots,y_r \rangle$ and
$\boldsymbol{z}= \langle z_1,\ldots,z_s \rangle$.  Consider the ring
of Laurent polynomials
\[\mathcal{R}:=
  K[x,y_1,y_1^{-1},\ldots,y_r,y_r^{-1},z_1,z_1^{-1},\ldots,
  z_s,z_s^{-1}] \, , \] which can be seen as a localisation\footnote{
  Recall that the \emph{localisation} of a commutative ring
  $\mathcal{U}$ in a multiplicatively closed subset $S$ such that
  $0_{\mathcal{U}}\not\in S$ is the ring of formal fractions
  $\mathcal{U}_S = \{ a/s : a\in\mathcal{U}, s\in S \}$, with addition
  and multiplication defined as usual.} of the polynomial ring
$\mathcal{A}:=K[x,y_1,\ldots,y_r,z_1,\ldots,z_s]$ in the
multiplicative set generated by the set of variables
$\{y_1,\ldots,y_r\} \cup \{z_1,\ldots,z_s\}$.  The multiplicative
units of $\mathcal{R}$ are the non-zero monomials in variables
$y_1,\ldots,y_r$ and $z_1,\ldots,z_s$.  As the localisation of a
unique factorisation domain, $\mathcal{R}$ is itself a unique
factorisation domain~\cite[Theorem 10.3.7]{Cohn02}.  From the proof of
this fact it moreover easily follows that $\mathcal{R}$ inherits from
$\mathcal{A}$ computability of factorisation into irreducibles (e.g.,
using the algorithm of~\cite{Lenstra87}).

We extend the operation of complex conjugation to a ring automorphism 
of $\mathcal{R}$ as follows.  Given a polynomial 
\[ P = \sum_{j=1}^n a_j x^{u_j} {y_1}^{v_{j1}} \ldots {y_r}^{v_{jr}}
{z_1}^{w_{j1}}\ldots {z_s}^{w_{js}} \, , \]
where $a_1,\ldots,a_n \in K$, define its conjugate to be
\[ \overline{P} := \sum_{j=1}^n \overline{a_j} x^{u_j} {y_1}^{v_{j1}}
  \ldots {y_r}^{v_{jr}} {z_1}^{-w_{j1}}\ldots {z_s}^{-w_{js}} \, . \]
This definition is motivated by thinking of the variables $x$ and
$y_1,\ldots,y_r$ as real-valued and the variables $z_1,\ldots,z_s$ as
taking values in the unit circle in the complex plane.

We will need the following proposition characterising those
polynomials in $P\in \mathcal{R}$ such that $P$ and $\overline{P}$ are
associates, i.e., such that $\overline{P}$ is equal to the product of
$P$ by a monomial.
Here we use pointwise notation for exponentiation:
given a tuple of integers
$\boldsymbol{u}=\langle u_1,\ldots,u_s \rangle$, we write
$\boldsymbol{z}^{\boldsymbol{u}}$ for the monomial
$z_1^{u_1}\ldots z_s^{u_s}$.

\begin{proposition}
  Let $P \in \mathcal{R}$ be such that
  $P=\boldsymbol{z}^{\boldsymbol{u}}\overline{P}$ for
  $\boldsymbol{u}\in\mathbb{Z}^s$.  Then either (i)~$P$ has the form
  $P=\boldsymbol{z}^{\boldsymbol{u}}Q$ for some $Q\in\mathcal{R}$ with
  $Q=\overline{Q}$, or (ii)~there exists $Q\in\mathcal{R}$ such that
  $P=Q+\boldsymbol{z}^{\boldsymbol{u}}\overline{Q}$ and $P$ does not
  divide $Q$ in $\mathcal{R}$.
\label{prop:three}
\end{proposition}
\begin{proof}
Consider a monomial $M$ such that $\boldsymbol{z}^{\boldsymbol{u}}\overline{M}=M$.  Then $M$ has a
  real coefficient and the exponent $\boldsymbol{w}$ of
  $\boldsymbol{z}$ in $M$ satisfies $2\boldsymbol{w}=\boldsymbol{u}$.
  Thus if $\boldsymbol{z}^{\boldsymbol{u}}\overline{M}=M$ for every monomial $M$ appearing in $P$
  then $P$ has the form $Q\boldsymbol{z}^{\boldsymbol{w}}$, where
  $2\boldsymbol{w}=\boldsymbol{u}$ and $Q$ is a polynomial in the
  variables $x$ and $\boldsymbol{y}$ with real coefficients.  In
  particular $Q=\overline{Q}$, and statement~(i) of the proposition
  applies.
  
  Suppose now that $\boldsymbol{z}^{\boldsymbol{u}}\overline{M}\neq M$ for some monomial $M$ appearing
  in $P$.  Then the map sending $M$ to $\boldsymbol{z}^{\boldsymbol{u}}\overline{M}$ induces a
  permutation of order 2 on the monomials on $P$.  Thus we may write
  $P=\sum_{j=1}^{n} M_j$, where $n=k+2\ell$ for some $k\geq 0$ and
  $\ell\geq 1$ such that $\boldsymbol{z}^{\boldsymbol{u}}\overline{M_j}=M_j$ for $1 \leq j \leq k$
  and $\boldsymbol{z}^{\boldsymbol{u}}\overline{M_j}=M_{j+\ell}$ for $k+1 \leq j \leq \ell$.  Then,
  writing
  $Q:=\frac{1}{2}\sum_{j=1}^k M_j + \sum_{j=k+1}^{k+\ell} M_j$, we
  have $P=Q+\boldsymbol{z}^{\boldsymbol{u}}\overline{Q}$.  

  The set of monomials appearing in $Q$ is a proper subset of the set
  of monomials appearing in $P$ (up to constant coefficients) and so
  $Q$ cannot be a constant multiple of $P$.  It also follows that for
  each variable
  $\sigma \in \{x,y_j,z_k : 1 \leq j \leq r, 1 \leq k \leq s\}$, the
  maximum degree of $\sigma$ in $P$ is at least its maximum degree in
  $Q$, and likewise for $\sigma^{-1}$.  Thus $Q$ cannot be a multiple
  of $P$ by a non-constant polynomial either.  We conclude that $P$
  does not divide $Q$.
\end{proof}

\subsection{Transcendence Theory}
We will use transcendence theory in our analysis of both the bounded
and unbounded variants of the Continuous Skolem Problem.  In the
unbounded case we will use the following classical result.
\begin{theorem}[Gelfond-Schneider]
  Let $a,b$ be algebraic numbers not equal to $0$ or $1$.  Then for
  any branch of the logarithm function, $\frac{\log(b)}{\log(a)}$ is
  either rational or transcendental.
\label{thm:GS}
\end{theorem}
In fact we will make use of the following corollary, which is obtained by applying Theorem~\ref{thm:GS} to the algebraic numbers
  $a=e^{i(\alpha_2-\alpha_1)}$ and $b=e^{i(\beta_2-\beta_1)}$.
\begin{corollary}
  Let $\alpha_1 \neq \beta_1$, $\alpha_2 \neq \beta_2$ all lie in
  $[0,\pi]$ and suppose that
  $\cos(\alpha_1),\cos(\alpha_2),\cos(\beta_1)$ and $\cos(\beta_2)$
  are algebraic.  Then $\frac{\beta_2-\alpha_2}{\beta_1-\alpha_1}$ is
  either rational or transcendental.
\label{corl:gelfond}
\end{corollary}

Our results in the bounded case depend on Schanuel's conjecture, a
unifying conjecture in transcendental number theory~\cite{lang},
which, if true, greatly generalises many of the central results in the
field (including the Gelfond-Schneider Theorem, above).  Recall that a
\emph{transcendence basis} of a field extension $L/K$ is a subset
$S \subseteq L$ such that $S$ is algebraically independent over $K$
and $L$ is algebraic over $K(S)$.  All transcendence bases of $L/K$
have the same cardinality, which is called the \emph{transcendence
  degree} of the extension.

\begin{conj}[Schanuel's Conjecture~\cite{lang}] \label{schanuel} Let
  $a_1,\ldots,a_n$ be complex numbers that are linearly independent
  over $\mathbb{Q}$.  Then the field
  $\mathbb{Q}(a_1,\ldots,a_n,e^{a_1},\ldots,e^{a_n})$ has
  transcendence degree at least $n$ over $\mathbb{Q}$.
\end{conj}
A special case of Schanuel's conjecture, that is known to hold
unconditionally, is the Lindemann-Weierstrass Theorem~\cite{lang}: if
$a_1,\ldots,a_n$ are algebraic numbers that are linearly independent
over $\mathbb{Q}$, then $e^{a_1},\ldots,e^{a_n}$ are algebraically
independent.

We apply Schanuel's conjecture via the following proposition.
\begin{proposition}
  Let $\{a_1,\ldots,a_r\}$ and $\{b_1,\ldots,b_s\}$ be
  $\mathbb{Q}$-linearly independent sets of real algebraic numbers.
  Furthermore, let $P,Q\in\mathcal{R}$ be two polynomials that have
  algebraic coefficients and are coprime in $\mathcal{R}$.  Then the
  equations
\begin{eqnarray}
P(t,e^{a_1 t},\ldots,e^{a_r t},e^{ib_1 t},\ldots,e^{i b_s t})&=& 0
\label{eq:poly1}
\\
Q(t,e^{a_1 t},\ldots,e^{a_r t},e^{ib_1 t},\ldots,e^{i b_s t})&=& 0
\label{eq:poly2}
\end{eqnarray}
have no non-zero common solution $t\in\mathbb{R}$.
\label{prop:two}
\end{proposition}
\begin{proof}
  Consider a solution $t\neq 0$ of Equations~(\ref{eq:poly1})
  and~(\ref{eq:poly2}).  By passing to suitable associates, we may
  assume without loss of generality that $P$ and $Q$ lie in
  $\mathcal{A}$, i.e., that all variables in $P$ and $Q$ appear with
  non-negative exponent.  Moreover, since $P$ and $Q$ are coprime in
  $\mathcal{R}$, their greatest common divisor $R$ in $\mathcal{A}$ is
  a monomial. In particular,
  \[R(t,e^{a_1 t},\ldots,e^{a_r t},e^{ib_1 t},\ldots,e^{i b_s t})\neq
  0 \, .\]
  Thus, dividing $P$ and $Q$ by $R$, we may assume that $P$ and $Q$
  are coprime in $\mathcal{A}$ and that Equations~(\ref{eq:poly1})
  and~(\ref{eq:poly2}) still hold.

  By Schanuel's conjecture, the extension
\[\mathbb{Q}(a_1t,\ldots,a_rt,ib_1t,\ldots,ib_st,e^{a_1t},\ldots,e^{a_rt},
e^{ib_1t},\ldots,e^{ib_st})/\mathbb{Q} \]
has transcendence degree at least $r+s$.  Since $a_1,\ldots,a_r$
and $b_1,\ldots,b_s$ are algebraic over $\mathbb{Q}$, writing 
\[ S:= \langle t,e^{a_1t},\ldots,e^{a_rt},e^{ib_1t},\ldots,e^{ib_st} \rangle \, , \]
it follows that the extension $\mathbb{Q}(S)/\mathbb{Q}$ also has 
transcendence degree at least $r+s$.

From Equations~(\ref{eq:poly1}) and~(\ref{eq:poly2}) we can regard $S$
as specifying a common root of $P$ and $Q$.  Pick some variable
$\sigma\in\{x,y_j,z_j : 1\leq i \leq r, 1 \leq j\leq s\}$ that has
positive degree in $P$.  Then the component of $S$ corresponding to
$\sigma$ is algebraic over the remaining components of $S$.  We claim
that the remaining components of $S$ are algebraically dependent and
thus $S$ comprises at most $r+s-1$ algebraically independent elements,
contradicting Schanuel's conjecture.  The claim clearly holds if
$\sigma$ does not appear in $Q$.  On the other hand, if $\sigma$ has
positive degree in $Q$ then, since $P$ and $Q$ are coprime
in $\mathcal{A}$, the multivariate resultant $\mathrm{Res}_\sigma(P,Q)$ is
a non-zero polynomial in the set of variables
$\{x,y_j,z_j : 1\leq i \leq r, 1 \leq j\leq s\} \setminus \{ \sigma
\}$
which has a root at $S$ (see, e.g.,~\cite[Page 163]{CoxLS07}).  Thus
the claim also holds in this case.  In either case we obtain a
contradiction to Schanuel's conjecture and we conclude that
Equations~(\ref{eq:poly1}) and~(\ref{eq:poly2}) have no non-zero
solution $t\in\mathbb{R}$.
\end{proof}

\section{Decidability of the Bounded Continuous Skolem Problem}
\label{sec:bounded}
Suppose that $\{a_1,\ldots,a_r\}$ and $\{ib_1,\ldots,ib_s\}$ are
$\mathbb{Q}$-linearly independent sets of real and imaginary numbers
respectively.  Let the ring of Laurent polynomials $\mathcal{R}$ be as
in Section~\ref{sec:laurent} and consider the exponential polynomial
\begin{gather}
f(t)=P(t,e^{a_1t},\ldots,e^{a_r t},e^{i b_1 t},\ldots,
e^{i b_s t}) \, , 
\label{eq:type1}
\end{gather}
where $P\in \mathcal{R}$ is irreducible.  We say that $f$ is a
\emph{Type-1} exponential polynomial if $P$ and $\overline{P}$ are not
associates in $\mathcal{R}$, we say that $f$ is \emph{Type-2} if
$P=\alpha\overline{P}$ for some $\alpha\in\mathbb{C}$, and we say that
$f$ is \emph{Type-3} if $P=U\overline{P}$ for some non-constant unit
$U\in\mathcal{R}$.  

\begin{example}
  The simplest example of a Type-3 exponential polynomial is
  $g(t)=1+e^{it}$.  Here $g(t)=P(e^{it})$, where $P(z)=1+z$ is an
  irreducible polynomial that is associated with its conjugate
  $\overline{P}(z)=1+{z^{-1}}$. Note that the exponential polynomial
  $f(t)=2+2\cos(t)$, which has infinitely many tangential zeros, factors as the product of two
  type-3 exponential polynomials $f(t)=g(t)\overline{g(t)}$.
\end{example}

In the case of a Type-2 exponential polynomial $P=\alpha \overline{P}$
it is clear that we must have $|\alpha|=1$.  Moreover, by replacing
$P$ by $\beta P$, where $\beta^2=\overline{\alpha}$, we may assume
without loss of generality that $P=\overline{P}$.  Similarly, in the
case of a Type-3 exponential polynomial, we can assume without loss of
generality that $P=\boldsymbol{z}^{\boldsymbol{u}}\overline{P}$ for
some non-zero vector $\boldsymbol{u}\in\mathbb{Z}^s$.

Now consider an arbitrary exponential polynomial
$f(t):=\sum_{j=1}^m P_j(t)e^{\lambda_j t}$.  Assume that the
coefficient field $K$ of $\mathcal{R}$ contains
the coefficients of $P_1,\ldots,P_m$.
Let $\{a_1,\ldots,a_r\}$ be a basis of the $\mathbb{Q}$-vector space
spanned by $\{ \Re(\lambda_j) : 1 \leq j \leq m\}$ and let
$\{b_1,\ldots,b_s\}$ be a basis of the the $\mathbb{Q}$-vector space
spanned by $\{ \Im(\lambda_j) : 1 \leq j \leq m\}$.  Without loss of
generality we may assume that each characteristic root $\lambda$ is an
integer linear combination of $a_1,\ldots,a_r$ and $ib_1,\ldots,ib_s$.
Then $e^{\lambda t}$ is a product of positive and negative powers of
$e^{a_1t},\ldots,e^{a_rt}$ and $e^{ib_1t},\ldots,e^{ib_st}$, and hence
there is a Laurent polynomial $P\in \mathcal{R}$ such that
\begin{gather}
f(t) = P(t,e^{a_1t},\ldots,e^{a_r t},e^{i b_1 t},\ldots,
e^{i b_s t}) \,  . 
\label{eq:eff}
\end{gather}

Since $P$ can be written as a product of irreducible factors, it
follows that $f$ can be written as product of Type-1, Type-2, and
Type-3 exponential polynomials, and moreover this factorisation can be
computed from $f$.  Thus it suffices to show how to decide the
existence of zeros of these three special forms of exponential
polynomial.  We will handle all three cases using Schanuel's
conjecture.

Writing the exponential polynomial $f(t)$ in (\ref{eq:eff}) in the
form $f(t)=\sum_{j=1}^m Q_j(t)e^{\lambda_j t}$, it follows from the
irreducibility of $P$ that the polynomials $Q_1,\ldots,Q_m$ have no
common root.  But then by the Lindemann-Weierstrass Theorem any zero
of $f$ must be transcendental (see~\cite[Theorem 8]{pisot}).

\begin{theorem}
  The Bounded Continuous Skolem Problem is decidable subject to
  Schanuel's conjecture.
\label{thm:main}
\end{theorem}
\begin{proof}
Consider an exponential polynomial
\begin{gather} 
f(t)=P(t,e^{a_1t},\ldots,e^{a_r t},e^{i b_1 t},\ldots,
e^{i b_s t}) \, , 
\label{eq:exp-poly}
\end{gather}
where $P\in\mathcal{R}$ is irreducible.  Suppose that
$\{a_1,\ldots,a_r\}$ and $\{ib_1,\ldots,ib_s\}$ are
$\mathbb{Q}$-linearly independent sets of, respectively, real and
imaginary numbers lying in the coefficient field $K$ of $\mathcal{R}$.
We show how to decide whether $f$ has a zero in a bounded interval
$I\subseteq \mathbb{R}_{\geq 0}$, considering separately the case of
Type-1, Type-2, and Type-3 exponential polynomials.

\subsection*{Case (i): $f$ is a type-1 exponential polynomial}
Note that $P$ and $\overline{P}$ are coprime in $\mathcal{R}$ since,
by assumption, they are both irreducible and are not associates.  We
claim that in this case the equation $f(t)=0$ has no solution
$t\in\mathbb{R}$.  Indeed $f(t)=0$ implies
\begin{eqnarray*}
P(t,e^{a_1t},\ldots,e^{a_r t},e^{i b_1 t},\ldots,e^{i b_s t})&=&0\\
\overline{P} (t,e^{a_1t},\ldots,e^{a_r t},e^{i b_1 t},\ldots,e^{i b_s t})&=&0 \, ,
\end{eqnarray*}
and the non-existence of a zero of $f$ follows
immediately from Proposition~\ref{prop:two}.

\subsection*{Case (ii): $f$ is a type-2 exponential polynomial}

In this case we have $P=\overline{P}$ and so $f$ is real-valued.  Our
aim is to use the procedure of Section~\ref{sec:zero-finding} to
determine whether or not $f$ has a zero in $[c,d]$, where $c,d\in\mathbb{Q}$. To this end,
notice first that $f(c),f(d)\neq 0$ since any root of $f$ must be
transcendental.  Moreover, since $f'$ is bounded on $[c,d]$, $f$ is
Lipschitz on $[c,d]$.  It remains to verify that the equations
$f(t)=0,f'(t)=0$ have no common solution $t\in [c,d]$.

We can write $f'(t)$ in the form
\begin{gather*}
f'(t) = Q(t,e^{a_1 t},\ldots,e^{a_r t}, e^{i b_1 t},\ldots,e^{i b_s t}) \, ,
\end{gather*}
where $Q$ is the polynomial 
\[ Q = \frac{\partial P}{\partial x} + \sum_{j=1}^r a_j y_j 
\frac{\partial P}{\partial y_j} + \sum_{j=1}^s i b_j z_j 
\frac{\partial P}{\partial z_j}
\, . \]
We claim that $P$ and $Q$ are coprime in $\mathcal{R}$.  Indeed, since
$P$ is irreducible, $P$ and $Q$ can only fail to be coprime if $P$
divides $Q$.

If $P$ has strictly positive degree $k$ in $x$ then $Q$ has degree
$k-1$ in $x$ and thus $P$ cannot divide $Q$.  (Recall that all
polynomials in $\mathcal{R}$ have non-negative degree in the variable
$x$.)  On the other hand, if $P$ has degree $0$ in $x$ then $Q$ is
obtained from $P$ by multiplying each monomial
$\boldsymbol{y}^{\boldsymbol{u}}\boldsymbol{z}^{\boldsymbol{v}}$
appearing in $P$ by the complex-number constant
$\sum_{j=1}^r a_j u_j + i \sum_{j=1}^s b_j v_j$.  Moreover, by the
assumption of linear independence of $\{a_1,\ldots,a_r\}$ and
$\{b_1,\ldots,b_s\}$, each monomial in $P$ is multiplied by a
different constant.  Since $P$ is not a unit, it has at least two
different monomials and so $P$ is not a constant multiple of $Q$.
Furthermore, for each variable
$\sigma\in \{ y_j,y_j^{-1} :1\leq j \leq r\} \cup \{z_j,z_j^{-1} :
1\leq j\leq s\}$, its degree in $P$ is equal to its degree in $Q$.
Thus $P$ cannot be a multiple of $Q$ by a non-constant polynomial
either.

We conclude that $P$ does not divide $Q$ and hence $P$ and $Q$ are
coprime.  It now follows from Proposition~\ref{prop:two} that the
equations $f(t)=f'(t)=0$ have no solution $t\in\mathbb{R}$.

\subsection*{Case (iii): $f$ is a type-3 exponential polynomial}

Suppose that $f$ is a Type-3 exponential polynomial.  Then in
(\ref{eq:exp-poly}) we have that
$P=\boldsymbol{z}^{\boldsymbol{u}}\overline{P}$ for some non-zero
vector $\boldsymbol{u}\in\mathbb{Z}^s$.  By
Proposition~\ref{prop:three} we can write $P = Q +
\boldsymbol{z}^{\boldsymbol{u}}\overline{Q}$ for some polynomial
$Q\in\mathcal{R}$ that is coprime with $P$.

Now define
\[
 g_1(t):= Q(t,e^{a_1t},\ldots,e^{a_r t},e^{i b_1 t},\ldots, e^{i b_s  t}) \]
and
$g_2(t):=e^{ib_1u_1}\cdots e^{ib_su_s}\overline{g_1(t)}$,
so that $f(t)=g_1(t)+g_2(t)$ for all $t$.

We show that $g_2(t)\neq 0$ for all $t\in\mathbb{R}$.  Indeed if $g_2(t)=0$ 
for some $t$ then we also have
$g_1(t)=0$ and hence $f(t)=0$.  For such a $t$ it follows that
\begin{eqnarray*}
P(t,e^{a_1t},\ldots,e^{a_r t},e^{i b_1 t},\ldots, e^{i b_s
    t}) &=& 0\\
Q(t,e^{a_1t},\ldots,e^{a_r t},e^{i b_1 t},\ldots, e^{i b_s
    t}) &=& 0 \, .
\end{eqnarray*}
But $P$ and $Q$ are coprime and so these two equations cannot both
hold by Proposition~\ref{prop:two}.  Not only do we have
$g_2(t)\neq 0$ for all $t\in\mathbb{R}$, but, applying the sampling
procedure in Section~\ref{sec:zero-finding} to $|g_2(t)|^2$ (which is
a differentiable function) we can compute a strictly positive lower
bound on $|g_2(t)|$ over the interval $[c,d]$.

Since $g_2(t)\neq 0$ for all $t\in\mathbb{R}$ we may define the function
$h:[c,d]\rightarrow \mathbb{R}$ by
\[ h(t):= \pi + i\log\left(\frac{g_1(t)}{g_2(t)}\right) \, .\]
Notice that $h(t)=0$ if and only if $f(t)=0$.  Our aim is to use the
procedure of Section~\ref{sec:zero-finding} to decide the existence of
a zero of $h$ in the interval $[c,d]$, and thus decide whether $f$ has
a zero in $[c,d]$.

Let $t \in (c,d)$ be such that $h(t)=0$.  Then $g_1(t)=-g_2(t)$ and so
$\frac{g_1(t)}{g_2(t)}=-1$ does not lie on the branch cut of the
logarithm function. It follows that $h$ is differentiable at $t$ and
\begin{eqnarray*}
 h'(t)=0
& \;\mbox{ iff }\; &
\frac{g_2(t)}{g_1(t)} \, \frac{g_1'(t)g_2(t)-g_2'(t)g_1(t)}{g_2(t)^2}  = 0\\[2pt]
& \;\mbox{ iff }\;& g_1'(t)g_2(t)-g_2'(t)g_1(t) = 0 \;\;\mbox{ (since
$|g_1(t)|=|g_2(t)|\neq 0$)}\\[2pt]
& \;\mbox{ iff }\; & g_1'(t)g_2(t)+g_2'(t)g_2(t)= 0
\;\;\mbox{ (since $g_1(t)=-g_2(t)$)}\\[2pt]
&  \;\mbox{ iff }\; & g_1'(t)+g_2'(t)=0 \\[2pt]
&  \;\mbox{ iff }\; & f'(t)=0 \, .
\end{eqnarray*}

Thus $h(t)=h'(t)=0$ implies $f(t)=f'(t)=0$.  But the proof in Case
(ii) shows that $f(t)=f'(t)=0$ is impossible.  (Nothing in that
argument hinges on $f$ being real-valued.)  Thus $h$ has no tangential
zeros in $(c,d)$.

We cannot directly use the procedure in Section~\ref{sec:zero-finding}
to decide whether $h$ has a zero in $[c,d]$ since $h$ is not
necessarily continuous: its value can jump from $-\pi$ to $\pi$ (or
\emph{vice versa}) due to the branch cut of the logarithm along the
positive real axis.  However, due to the strictly positive lower bound
on $|g_2(t)|$, the function $|h|$ \emph{is} Lipschitz on $[c,d]$.
Thus, applying the sampling procedure in
Section~\ref{sec:zero-finding} for computing lower and upper bounds of
Lipschitz functions we can compute a set $E\subseteq [c,d]$ such that
$E$ is a finite union of intervals with rational endpoints,
$|f(t)|\leq \frac{2\pi}{3}$ for $t\in E$, and
$|f(t)|\geq \frac{\pi}{3}$ for $t\not\in E$.  In particular, $E$
contains all zeros of $f$ in $[c,d]$ and $f$ is Lipschitz on $E$.
Thus we can apply the zero-finding procedure from
Section~\ref{sec:zero-finding} to the restriction $h\!\restriction E$ and
thereby decide whether or not $h$ has a zero on $[c,d]$.

\end{proof}

\section{The Unbounded Case}
In this section we consider the unbounded case of the Continuous
Skolem Problem.  For our analysis it is convenient to present
exponential polynomials in the form
\begin{gather} 
f(t) = \sum_{j=1}^n
e^{r_jt} \left(P_{1,j}(t)\cos(\omega_j t) +
  P_{2,j}(t) \sin(\omega_j t)\right) \, ,
\label{eq:freq-form}
\end{gather}
where $r_j,\omega_j$ are real algebraic numbers and $P_{1,j}, P_{2,j}$
are polynomials with real algebraic coefficients for $j=1,\ldots,n$.
Our aim is to classify the difficulty of the problem
in terms of the number of rationally linear independent frequencies
$\omega_1,\ldots,\omega_n$.

Recall that in Section~\ref{sec:bounded} we have shown the bounded
problem to be decidable subject to Schanuel's Conjecture.  In 
Appendix~\ref{sec:one-osc} we give a reduction of the
unbounded problem to the bounded problem in case the set of
frequencies spans a one-dimensional vector space over $\mathbb{Q}$.
In the present section we give a reduction of the unbounded problem to
the bounded problem in case the set of frequencies spans a
two-dimensional vector space over $\mathbb{Q}$ and the polynomials
$P_{1,j}$ and $P_{2,j}$ are all constant.  (This last condition is
equivalent to the assumption that $f(t)$ is simple.)  The argument in
the two-dimensional case is a more sophisticated version of that in
the one-dimensional case, although the result is not more general due
the assumption of simplicity.

In Appendix~\ref{sec:hardness} we present a family of
instances showing that obtaining decidability of the unbounded problem
in the two-dimensional case without the assumption of simplicity would
require much finer Diophantine-approximation bounds than are currently
known.

\subsection{Background on Semi-Algebraic Sets}
A subset of $\mathbb{R}^n$ is \emph{semi-algebraic} if it is defined
by a Boolean combination of constraints of the form
$P(x_1,\ldots,x_n) > 0$, where $P$ is a polynomial with real algebraic
coefficients.  A partial function
$f:\mathbb{R}^n\rightarrow\mathbb{R}$ is semi-algebraic if its graph
is a semi-algebraic subset of $\mathbb{R}^{n+1}$.  The
Tarski-Seidenberg Theorem~\cite[Section 1]{Bierstone88} states that the semi-algebraic sets are
closed under projection and are therefore precisely the first-order
definable sets over the structure $(\mathbb{R},<,+,\cdot,0,1)$.

Let $(i_1,\ldots,i_n)$ be a sequence of zeros and ones of length $n\geq 1$.
An \emph{$(i_1,\ldots,i_n)$-cell} is a subset of $\mathbb{R}^n$,
defined by induction on $n$ as follows:
\begin{enumerate}
\item[(i)] A $(0)$-cell is a singleton subset of $\mathbb{R}$ and 
a $(1)$-cell is an open interval $(a,b)\subseteq\mathbb{R}$.
\item[(ii)] Let $X\subseteq\mathbb{R}^n$ be a $(i_1,\ldots,i_n)$-cell
  and $f:X\rightarrow\mathbb{R}$ a continuous semi-algebraic function.
  Then
  $\{(\boldsymbol{x},f(\boldsymbol{x}))\in\mathbb{R}^{n+1}
  :\boldsymbol{x}\in X\}$
  is a $(i_1,\ldots,i_n,0)$-cell, while
  $\{(\boldsymbol{x},y)\in\mathbb{R}^{n+1} : \boldsymbol{x}\in X
  \wedge y<f(\boldsymbol{x})\}$
  and
  $\{(\boldsymbol{x},y)\in\mathbb{R}^{n+1} : \boldsymbol{x}\in X
  \wedge y>f(\boldsymbol{x})\}$ are both $(i_1,\ldots,i_n,1)$-cells.
\item[(iii)] Let $X\subseteq\mathbb{R}^n$ be a $(i_1,\ldots,i_n)$-cell
  and $f,g:X\rightarrow\mathbb{R}$ continuous semi-algebraic
  functions such that $f(\boldsymbol{x})<g(\boldsymbol{x})$ for all
  $\boldsymbol{x}\in X$.  Then
  $\{ (\boldsymbol{x},y) \in \mathbb{R}^{n+1} : f(\boldsymbol{x})<y<
  g(\boldsymbol{x})\}$ is a $(i_1,\ldots,i_n,1)$-cell.
\end{enumerate}
A cell in $\mathbb{R}^n$ is a $(i_1,\ldots,i_n)$-cell for some
(necessarily unique) sequence $(i_1,\ldots,i_n)$.

A fundamental result about semi-algebraic sets, that we will use below,
is the Cell-Decomposition Theorem~\cite{BasuPR06}: given a
semi-algebraic set $E\subseteq \mathbb{R}^n$ one can compute a
partition of $E$ as a disjoint union of cells
$E=C_1\cup \ldots \cup C_m$.

We will also need the following result, proved in Appendix~\ref{sec:missing}.
\begin{lemma}
  Let $D\subseteq \mathbb{R}^{n}$ be a semi-algebraic set,
  $g : D \rightarrow \mathbb{R}$ a bounded semi-algebraic function, and
  $r_1,\ldots,r_n$ real algebraic numbers.  Define
  $S=\{t\in\mathbb{R}_{\geq 0} : (e^{r_1t},\ldots,e^{r_nt}) \in D \}$.
  Then
\begin{enumerate}
\item[(i)] It is decidable whether or not $S$ is bounded.  If
  $S$ is bounded then we can compute $T_0\in\mathbb{N}$ such that
  $S\subseteq [0,T_0]$ and if $S$ is unbounded then we can compute
  $T_0\in\mathbb{N}$ such that $(T_0,\infty) \subseteq S$.
\item[(ii)] If $S$ is unbounded then 
the limit $g^* = \lim_{t\rightarrow \infty} g(e^{r_1t},\ldots,e^{r_nt})$
  exists, is an algebraic number, and there are effective constants
  $T_1,\varepsilon>0$ such that
  $|g(e^{r_1t},\ldots,e^{r_nt})- g^*| < e^{-\varepsilon t}$ for all $t>T_1$.
\end{enumerate}
\label{lem:combine}
\end{lemma}

\subsection{Two Linearly Independent Frequencies}
The following lemma, which is a reformulation of \cite[Lemma
13]{pisot}, plays an instrumental role in this section.  The lemma
itself relies on a powerful quantitative result in transcendence
theory---Baker's Theorem on linear forms in logarithms of algebraic
numbers~\cite{baker}.
\begin{lemma}\label{lem:twoCosBaker}
Let $b_1,b_2$ be real algebraic numbers,
  linearly independent over $\rats$.  Furthermore, let
  $\varphi_1,\varphi_2$ be real numbers such that $e^{i\varphi_1}$ and
  $e^{i\varphi_2}$ are algebraic.  Then there exist effectively
  computable constants $N,T>0$ such that for all $t\geq T$ and all
  $k_1,k_2 \in \mathbb{Z}$, at least one of
  $|b_1t-\varphi_1-2k_1\pi| > 1/t^N$ and
  $|b_2t-\varphi_2-2k_2\pi| > 1/t^N$ holds.
\end{lemma}

 The main result of the section is the following.
\begin{theorem}\label{thm:mtargument2}
  Let
  $f(t) = \sum_{j=1}^n e^{r_jt} \left(a_{1,j} \cos(\omega_j t) +
    a_{2,j} \sin(\omega_j t)\right)$
  be an exponential polynomial where $r_j,a_{1,j},a_{2,j},\omega_j$
  are real algebraic numbers and the $\mathbb{Q}$-span of
  $\{ \omega_1, \ldots , \omega_n \}$ has dimension two as a
  $\mathbb{Q}$-vector space.  Then we can decide whether or not
  $\{t\in \mathbb{R}_{\geq 0}:f(t)=0\}$ is bounded and, if bounded, we
  can compute an integer $T$ such that
  $\{t\in \mathbb{R}_{\geq 0}:f(t)=0\}\subseteq [0,T]$.
\end{theorem}

\begin{proof}
  Let $b_1,b_2$ be real algebraic numbers, linearly independent over
  $\mathbb{Q}$, such that $\omega_j$ is an integer linear combination
  of $b_1$ and $b_2$ for $j=1,\ldots,n$.  For each $n\in \mathbb{Z}$,
  $\sin(nb_1t)$ and $\cos(nb_1t)$ can be written as polynomials in
  $\sin(b_1t)$ and $\cos(b_1t)$ with integer coefficients, and
  similarly for $b_2$.  It follows that we can write $f$ in the form
\[ f(t) = Q(e^{r_1 t},\ldots, e^{r_nt},\cos(b_1t),\sin(b_1t),
\cos(b_2t),\sin(b_2t)) \]
for some polynomial $Q$ with real algebraic coefficients that
is computable from $f$.

Write
$R_{++} = \{t\geq 0:\sin(b_1t)\geq 0 \wedge \sin(b_2t)\geq 0 \}$,
$R_{+-} = \{t\geq 0: \sin(b_1t)\geq 0 \wedge \sin(b_2t)\leq 0 \}$, and
likewise define $R_{- +}$, $R_{- -}$ for the two remaining sign
conditions on $\sin(b_1t)$ and $\sin(b_2t)$. We show how to decide boundedness of
$\{ t\in R_{++} :f(t)=0\}$.  (The cases for $R_{+ -}$, $R_{- +}$, and
$R_{- -}$ follow \emph{mutatis mutandis}.)  The idea is to compute a
partition of $\{ t \in R_{+ +} : f(t)=0\}$ into components
$Z_1,\ldots,Z_m$ and to separately decide boundedness of each component $Z_j$.

Define a semi-algebraic set 
\[ E = \big\{(\boldsymbol{u},x_1,x_2) \in \mathbb{R}^{n+2}: \exists
y_1,y_2\geq 0 \left(x_1^2+y_1^2=x_2^2+y_2^2=1 \wedge
  Q(\boldsymbol{u},x_1,y_1,x_2,y_2) = 0 \right)\big \} \, . \]
Then for $t\in R_{+ +}$ we have $f(t)=0$ if and only if
$(e^{\vr t},\cos(b_1t),\cos(b_2t)) \in E$, where $\vr = (r_1,\ldots,r_n)$. 
Now consider a cell
decomposition $E=C_1 \cup \ldots \cup C_m$ for cells
$C_1,\ldots,C_m \subseteq \mathbb{R}^{n+2}$, and define
\begin{gather} Z_j = \{ t\in R_{+ +} : (e^{\vr
  t},\cos(b_1t),\cos(b_2t)) \in C_j \}\, , \qquad j=1,\ldots,m,
\label{def:Zj}
\end{gather}
Then $\{ t \in R_{+ +} : f(t)=0\} =
Z_1 \cup \ldots \cup Z_m$.

Now fix $j\in\{1,\ldots,m\}$.  We show how to decide boundedness of
$Z_j$.  To this end, write $D_j \subseteq \mathbb{R}^{n}$ for the
projection of the corresponding cell $C_j \subseteq \mathbb{R}^{n+2}$
on the first $n$ coordinates.

First suppose that $\{ t \in \mathbb{R} : e^{\vr t} \in D_j \}$ is
bounded.  Then by Lemma~\ref{lem:combine} we can compute an
upper bound $T$ of this set.  But $Z_j \subseteq \{ t \in
\mathbb{R}_{\geq 0} : e^{\vr t} \in D_j\}$ and so $Z_j \subseteq
       [0,T]$.

On the other hand, suppose that $\{ t \in \mathbb{R} : e^{\vr t} \in
D_j \}$ is unbounded.  Then, by Lemma~\ref{lem:combine}, this
set contains an unbounded interval $(T,\infty)$ for some
$T\in\mathbb{N}$.  Write $I=[-1,1]$ and define functions
$g_1,g_2,h_1,h_2 : D_j \rightarrow \mathbb{R}$ by
\begin{align}
\label{eq:boundary1}
g_1(\boldsymbol{u}) & =  \inf \{ x\in I : \exists y \, (\boldsymbol{u},x,y) \in C_j\} \qquad
  & g_2(\boldsymbol{u}) &  =  \inf \{ y\in I  : \exists x \, (\boldsymbol{u},x,y) \in C_j\} \\
h_1(\boldsymbol{u}) & =  \sup \{ x\in I: \exists y \, (\boldsymbol{u},x,y) \in C_j\} \qquad
 & h_2(\boldsymbol{u}) & =  \sup \{ y\in I : \exists x \, (\boldsymbol{u},x,y) \in C_j\} 
\label{eq:boundary2}
\end{align}
  These functions are all semi-algebraic by quantifier elimination.  Hence
  by Lemma~\ref{lem:combine} the limits
  $g_i^* = \lim_{t\rightarrow \infty} g_i(e^{\vr t})$ and
  $h_i^* = \lim_{t\rightarrow \infty} h_i(e^{\vr t})$ exist for $i=1,2$ and are
  algebraic numbers.  Clearly we have $g_1^*\leq h_1^*$ and
  $g_2^*\leq h_2^*$.  We now consider three cases according to the
  strictness of these inequalities.

\subsubsection{Case I: $g_1^*= h_1^*$ and $g_2^*=h_2^*$.}  
We show that $Z_j$ is bounded and that we can
compute $T_2$ such that $Z_j\subseteq [0,T_2]$.

By Lemma~\ref{lem:combine} there exist $T_1,\varepsilon>0$ such
that for all $t>T_1$ and $i=1,2$, 
\begin{gather}
  |g_i(e^{\vr t}) - g_i^*| < e^{-\varepsilon t} \mbox{ and }
  |h_i(e^{\vr t}) - h_i^*| < e^{-\varepsilon t} \, .
\label{eq:bounds}
\end{gather}

Then for $t \in R_{+ +}$ such that $t>T_1$ we have
\begin{eqnarray}
t\in Z_j & \Longleftrightarrow &
\left(e^{\vr t},\cos(b_1t),\cos(b_2t)\right) \in C_j \;\;\mbox{ (by (\ref{def:Zj}))}\notag \\
       & \Longrightarrow & 
       g_1(e^{\vr t}) \leq \cos(b_1t) \leq h_1(e^{\vr t}) \,\mbox{ and }\,
       g_2(e^{\vr t}) \leq \cos(b_2t) \leq h_2(e^{\vr t}) 
\;\;\mbox{ (by (\ref{eq:boundary1})(\ref{eq:boundary2}))}
\notag \\
&\Longrightarrow &
\left| \cos(b_1t) - g_1^* \right| < e^{-\varepsilon t} \mbox{ and }
\left| \cos(b_2t) - g_2^* \right| < e^{-\varepsilon t} 
\;\;\mbox{ (by (\ref{eq:bounds}))}  \label{eq:just}
\end{eqnarray}

Write $g_1^* = \cos(\varphi_1)$ and $g_2^*=\cos(\varphi_2)$ for some
$\varphi_1,\varphi_2\in [0,\pi]$.  Since
$|\cos(\varphi_1+x)-\cos(\varphi_1)| \geq x^3/3$ for all $x$
sufficiently small (by a Taylor expansion), the inequality
(\ref{eq:just}) implies that for some $k_1,k_2 \in \mathbb{Z}$,
\begin{gather} |b_1t - \varphi_1 - 2k_1\pi| < 3e^{-\varepsilon t/3} \,\mbox{ and }\,
|b_2t - \varphi_2 - 2k_2\pi| < 3e^{-\varepsilon t/3} \, .
\label{eq:upper}
\end{gather}
Combining the upper bounds in (\ref{eq:upper}) with the polynomial lower
bounds $|b_1t-\varphi_1-2k_1\pi| > 1/t^N$ and
$|b_2t-\varphi_2-2k_2\pi| > 1/t^N$ from Lemma~\ref{lem:twoCosBaker} we
obtain an effective bound $T_2$ for which $t\in Z_j$ implies $t<T_2$.

\subsubsection{Case II: $g_1^*< h_1^*$.}  In this case we show that $Z_j$
is unbounded.  The geometric intuition is as follows.  We imagine a
particle in the plane whose position at time $t$ is
$(\cos(b_1t),\cos(b_2t))$, together with a ``moving target'' whose
extent at time $t$ is $\Gamma_t=\{(x,y) : (e^{\vr t},x,y) \in C_j \}$.
Below we essentially argue that such a particle is bound to hit
$\Gamma_t$ at some time $t$ since its orbit is dense in $[-1,+1]^2$ and
$\Gamma_t$ has positive dimension in the limit.

Proceeding formally, first notice that $C_j$ cannot be a
$(\ldots,0,1)$-cell or a $(\ldots,0,0)$-cell, for then we would have
$g_1(\boldsymbol{u}) = h_1(\boldsymbol{u})$ for all $\boldsymbol{u}\in
D_j$ and hence $g_1^*=h_1^*$.  Thus $C_j$ must either be a
$(\ldots,1,0)$-cell or a $(\ldots,1,1)$-cell.  In either case, $C_j$
includes a cell of the form
$\{(\boldsymbol{u},x,\xi(\boldsymbol{u},x)) : \boldsymbol{u}\in D,
g_1(\boldsymbol{u}) < x < h_1(\boldsymbol{u}) \}$ for some
semi-algebraic function $\xi$.

Let $c,d$ be real algebraic numbers such that $g_1^*<c<d<h_1^*$.
Write $c=\cos(\psi')$ and $d=\cos(\psi)$ for $0\leq \psi<\psi' \leq
\pi$.  By Lemma~\ref{lem:combine} the limits
$\lim_{t\rightarrow\infty} \xi(e^{\vr t},c)$ and
$\lim_{t\rightarrow\infty} \xi(e^{\vr t},d)$ exist and are algebraic
numbers in the interval $[-1,1]$.  Let $\theta,\theta' \in [0,\pi]$ be
such that $\cos(\theta)=\lim_{t\rightarrow\infty} \xi(e^{\vr t},d)$
and $\cos(\theta') = \lim_{t\rightarrow\infty} \xi(e^{\vr t},c)$.

By Corollary~\ref{corl:gelfond} we know that
$\frac{\theta'-\theta}{\psi'-\psi}$ is either rational or
transcendental.  In particular we know that it is not equal to
$\frac{b_2}{b_1}$, which is algebraic and irrational.  Let us
suppose that $\frac{\theta'-\theta}{\psi'-\psi} > \frac{b_2}{b_1}$
(the converse case is almost identical).  Then there exists $\theta''$
with $\theta<\theta''<\theta'$, such that
\begin{gather}
\theta< \theta'' + \frac{b_2}{b_1}(\psi'-\psi) < \theta' \, . 
\label{eq:limit} 
\end{gather}

Since $2\pi,b_1,b_2$ are linearly independent over $\mathbb{Q}$ it
follows from Kronecker's approximation theorem that
$\{ (b_1t,b_2t)\bmod 2\pi : t \in \mathbb{R}_{\geq 0} \}$ is dense in
$[0,2\pi)^2$ (see~\cite[Chapter 23]{HardyW38}).  Thus there
is an increasing sequence $t_1<t_2< \ldots$, with
$b_1t_n \equiv \psi \bmod 2\pi$ for all $n$, such that $b_2t_n \bmod 2\pi$
converges to $\theta''$.  Then, defining
$s_1<s_2< \ldots$ by $s_n = t_n + \frac{\psi'-\psi}{b_1}$, we have
$b_1s_n \equiv \psi' \bmod 2\pi$ for all $n$ and, by (\ref{eq:limit}),
\[ 
\lim_{n \rightarrow \infty} b_2s_n  =
\lim_{n\rightarrow \infty} b_2t_n +\frac{b_2}{b_1}(\psi'-\psi) =
\theta'' + \frac{b_2}{b_1}(\psi'-\psi) < \theta' \quad(\bmod\; 2\pi) \]

Let $\eta(t)=\xi(e^{\vr t},\cos(b_1t)) - \cos(b_2t)$.  Then
for $t\in R_{++}$ such that $g(e^{\vr t})<\cos(b_1t)<h(e^{\vr t})$,
\begin{eqnarray*}
\eta(t)=0&\Longrightarrow & \cos(b_2t)=\xi(e^{\vr t},\cos(b_1t))\\
         &\Longrightarrow & (e^{\vr t},\cos(b_1t),\cos(b_2t)) \in C_j\\
&\Longrightarrow & t\in Z_j \mbox{ (by (\ref{def:Zj}))} \, .
\end{eqnarray*}
Now
$\lim_{n\rightarrow \infty}\eta(t_n)=\cos(\theta)-\cos(\theta'')>0$
and
$\lim_{n\rightarrow \infty} \eta(s_n)<\cos(\theta')-\cos(\theta')=0$.
Moreover for $n$ sufficiently large we have
$[t_n,s_n]\subseteq R_{++}$.  It follows that $\eta(t)$ has a
zero in every interval $[t_n,s_n]$ for $n$ large enough.  We conclude
that $Z_j$ is unbounded.

\subsubsection{Case III: $g_2^* < h_2^*$.}  This case is symmetric to
Case II and we omit details.
\end{proof}
\subparagraph*{Acknowledgements} The authors wish to thank Angus
Macintyre for helpful comments.  Ventsislav Chonev is supported by
Austrian Science Fund (FWF) NFN Grant No S11407-N23 (RiSE/SHiNE), ERC
Start grant (279307: Graph Games), and ERC Advanced Grant (267989:
QUAREM).  Jo\"{e}l Ouaknine is supported by ERC grant AVS-ISS
(648701).  James Worrell is supported by EPSRC grant
EP/N008197/1.

\bibliographystyle{abbrv}
\bibliography{bounded}

\newpage
\appendix

\section{Hardness}
\label{sec:hardness}
In this section we show that decidability of the Continuous Skolem
Problem entails significant new effectiveness results in Diophantine
approximation, thereby identifying a formidable mathematical obstacle
to further progress in the unbounded case.

Diophantine approximation is a branch of number theory concerned with
approximating real numbers by rationals.  A central role is played in
this theory by the notion of \emph{continued fraction expansion},
which allows to compute a sequence of rational approximations to a
given real number that is optimal in a certain well-defined sense.
For our purposes it suffices to note that the behaviour of the
simple continued fraction expansion of a real number $a$ is closely related
to the \emph{(homogeneous Diophantine approximation) type} of $a$,
which is defined to be
\[
L(a) \defn \inf\left\{ c : \left| a - \frac{n}{m} \right| < \frac{c}{m^2}
\mbox{ for some $m,n\in\mathbb{Z}$}\right\}  \, .
\]
Let $[n_1,n_2,n_3,\ldots]$ be the sequence of partial quotients in the 
simple continued fraction expansion of $a$.
Then, writing $K(a):=\sup_{k\geq 0} n_k$,
it is shown in~\cite[pp.\ 22-23]{Schmidt80} that $L(a)=0$ if and only
if $K(a)$ is infinite and otherwise
\[ K(a) \leq L(a)^{-1} \leq K(a)+2 \, .\]

It is well known that a real number algebraic number of degree two
over the rationals has a simple continued fraction expansion that is
ultimately periodic. In particular, such numbers have bounded partial
quotients.  But nothing is known about real algebraic numbers of
degree three or more---no example is known with bounded partial
quotients, nor with unbounded quotients. Guy~\cite{Guy04} asks:
\begin{quote}
  \emph{Is there an algebraic number of degree greater than two whose simple
  continued fraction expansion has unbounded partial quotients?  Does
  every such number have unbounded partial quotients?}
\end{quote}
In other words, the question is whether there is a real algebraic
number $a$ of degree at least three such that $L(a)$ is strictly
positive, or whether $L(a)=0$ for all such $a$.

Recall that a real number $x$ is \emph{computable} if there is an
algorithm which, given any rational $\varepsilon>0$ as input, returns
a rational $q$ such that $|q-x|<\varepsilon$.  The main result of this
section is Theorem~\ref{thm:hard}, which shows that the existence of a
decision procedure for the general Continuous Skolem Problem entails
the computability of $L(a)$ for all real algebraic numbers $a$.  Now
one possibility is that all such numbers $L(a)$ are zero, and hence
trivially computable.  However the significance of
Theorem~\ref{thm:hard} is that in order to prove the decidability of
the Continuous Skolem Problem one would have to establish, \emph{one
  way or another}, the computability of $L(a)$ for every real
algebraic number $a$.

Fix positive real algebraic $a, c$ and define the functions:
\begin{eqnarray*}
f_1(t) & = e^t(1-\cos(t)) + t (1 - \cos(a t)) - c\sin(a t), & \\
f_2(t) & = e^t(1-\cos(t)) + t (1 - \cos(a t)) + c\sin(a t), & \\
f(t) & = e^t(1-\cos(t)) + t (1 - \cos(a t)) - c|\sin(a t)|  & = 
\min\{f_1(t),f_2(t)\}.
\end{eqnarray*}
Then $f_1(t)$ and $f_2(t)$ are exponential polynomials.  Moreover it
is easy to check that the function $f(t)$ has a zero in an interval of
the form $(T,\infty)$ if and only if at least one of $f_1(t)$,
$f_2(t)$ has a zero in $(T,\infty)$.  

We will first prove two lemmas which show a connection between the
existence zeros of $f(t)$ and the type $L(a)$.  We then will
derive an algorithm to compute $L(a)$ using an oracle for the
Continuous Skolem Problem, thereby demonstrating our desired
hardness result.

\begin{lemma}\label{hardForward}
Fix real algebraic $a, c$ and $\varepsilon\in\mathbb{Q}$ with $a,c>0$
and $\varepsilon\in(0, 1)$. There exists an effective threshold $T$,
dependent on $a, c, \varepsilon$, such that if $f(t)=0$ for some
$t\geq T$, then $L(a)\leq c/2\pi^2(1-\varepsilon)$.
\end{lemma}
\begin{proof}
Suppose $f(t) = 0$ for some $t\geq T$. Define $\delta_1= t- 2\pi m$ 
and $\delta_2 = a t - 2\pi n$, where $m, n \in \mathbb{N}$ and 
$\delta_1,\delta_2\in [-\pi,\pi)$.  Then we have
\[
\left| a - \frac{n}{m} \right| = 
\frac{|\delta_2-a \delta_1|}{2\pi m}.
\]
We will show that for $T$ chosen large enough, if $f(t)=0$ for
$t\geq T$ then we can
bound $|\delta_2|$ and $|a\delta_1|$ separately from above
and then apply the triangle inequality to bound $|\delta_2 -a \delta_1|$, 
obtaining the desired upper bound on $L(a)$.  

Define $0<\alpha<1$ by $\alpha^2 = (1-\varepsilon^2)$.
Since $\displaystyle m\geq \frac{t-\pi}{2\pi} \geq \frac{T-\pi}{2\pi}$, for sufficiently large $T$ we have 
\begin{equation}\label{hardprop1}
t \geq 2\pi (m - 1) \geq 2\pi m\alpha \, .
\end{equation}
Furthermore, since $\alpha x^2/2 \leq 1-\cos(x)$ for $|x|$ sufficiently small, we may assume that $T$ is large enough such that the following is valid for $|x|\leq \pi$:
\begin{equation}\label{hardprop2}
\mbox{if $1-\cos(x)\leq c\pi /T$ then
$\alpha x^2/2 \leq 1-\cos(x)$.}
\end{equation}

We have the following chain of inequalities, where $(\ast)$ follows 
from $f(t)=0$ and $e^t(1-\cos(t))\geq 0$:
\begin{gather*}
1-\cos(\delta_2) = 1-\cos(at) \stackrel{(\ast)}{\leq} \frac{c|\sin(at)|}{t}
   = \frac{c|\sin(\delta_2)|}{t} \leq \frac{c|\delta_2|}{t} \, .
\label{eq:upper*}
\end{gather*}
It follows that $1-\cos(\delta_2) \leq c\pi/t$ and so by
   (\ref{hardprop2}) we also have
\begin{gather*}
\frac{\alpha \delta_2^2}{2} \leq 1-\cos(\delta_2) \, .
\label{eq:lower}
\end{gather*}
Combining the upper and lower bounds on $1-\cos(\delta_2)$ 
and using (\ref{hardprop1}), we have
\[
|\delta_2| \leq \frac{2c}{\alpha t} \leq 
\frac{2c}{2\pi m\alpha ^2} =
\frac{c}{m\pi(1-\varepsilon^2)}.
\]

We next seek an upper bound on $|\delta_1|$.  To this end,
let $T$ be large enough so that
\begin{equation}\label{hardprop3}
ce^{-t} \leq \left(\frac{c\varepsilon}{2a\alpha t}\right)^2 \;\mbox{ 
for $t\geq T$.}
\end{equation}
Then the following chain of inequalities holds:
\begin{align*}
\frac{\delta_1^2}{16} & \;\leq\; 1-\cos(\delta_1) 
& \mbox{ \{ valid for all $|\delta_1| \leq \pi$ \} } \\[2pt]
& \;=\; \frac{c|\sin(\delta_2)| - t(1-\cos(\delta_2))}{e^t} 
& \mbox{ \{ since $f(t) = 0$ \} } \\[2pt]
& \;\leq\; ce^{-t} 
& \mbox{ \{ since $|\sin(\delta_2)|, |\cos(\delta_2)| \leq 1$\}} \\ 
& \;\leq\; \left(\frac{c\varepsilon}{2a\alpha t}\right)^2 
&  \mbox{ \{ by (\ref{hardprop3}) \}} \\[2pt]
& \;\leq\; \left(\frac{c\varepsilon}{4a\pi \alpha^2 m}\right)^2 
& \mbox{ \{ by (\ref{hardprop1}) \}}
\end{align*}
It follows that
\[ |a\delta_1| \leq \frac{c\varepsilon}{\pi m(1-\varepsilon^2)} \, .\]

Finally, by the triangle inequality and the bounds on $|a \delta_1|$
and $|\delta_2|$,
we have
\[
\left| a - \frac{n}{m} \right| = \frac{|\delta_2 - a\delta_1|}{2\pi m}
\leq \frac{|\delta_2| + |a \delta_1|}{2\pi m} \leq 
\frac{c+c\varepsilon}{2\pi^2m^2(1-\varepsilon^2)} =
\frac{c}{2\pi^2 m^2(1-\varepsilon)} \, ,
\]
so the natural numbers $n, m$ witness $L(a)\leq c/2\pi^2(1-\varepsilon)$.
\end{proof}

\begin{lemma}
Fix real algebraic $a, c$ and $\varepsilon\in\mathbb{Q}$ with $a,c>0$ and
$\varepsilon\in(0, 1)$. There exists an effective threshold $M$,
dependent on $a, c, \varepsilon$, 
such that if $L(a) \leq c(1-\varepsilon)/2\pi^2$ holds and is witnessed
by natural numbers $n, m$ with $m \geq M$, then $f(t) = 0$ for some
$t \geq 2\pi M$.
\end{lemma}
\begin{proof}
Select $M$ large enough, so that $c(1-\varepsilon)/\pi M < \pi$ and
\begin{equation}\label{hardprop5}
\mbox{if $|x|< c(1-\varepsilon)/\pi M$, then 
$(1-\varepsilon)|x|\leq|\sin(x)|$.}
\end{equation}
Suppose now that $L(a) \leq c(1-\varepsilon)/2\pi^2$, let this be
witnessed by $n, m \in\mathbb{N}$ with $m\geq M$ and
define $t \defn 2\pi m$. We will show that $f(t) \leq 0$. This suffices,
because $f(t)$ is continuous and moreover is positive for arbitrarily
large times, so it must have a zero on $[t, \infty)$.

Since $L(a) \leq c(1-\varepsilon)/2\pi^2$, we have 
$|am-n| \leq c(1-\varepsilon)/2\pi^2m$. Therefore, we can write
$at = 2\pi a m = 2\pi n + \delta$ for some $\delta$ satisfying 
$|\delta| \leq c(1-\varepsilon)/\pi m < \pi$. We have
\begin{align*}
& f(t) & \\
= & \mbox{ \{ as $\cos(t)=1$ \} } & \\
& t(1-\cos(\delta)) - c|\sin(\delta)| & \\
\leq & \mbox{ \{ by (\ref{hardprop5}) and $1-\cos(x)\leq x^2/2$ \} } & \\
& \pi m\delta^2 - c(1-\varepsilon)|\delta| & \\
\leq & \mbox{ \{ by $|\delta|\leq c(1-\varepsilon)/\pi m$ \} } & \\
& 0. &
\end{align*}
\end{proof}
The following corollary is immediate:
\begin{lemma}\label{hardBackward}
Fix real algebraic $a, c$ and $\varepsilon\in\mathbb{Q}$ with $a,c>0$ and
$\varepsilon\in(0, 1)$. There exists an effective threshold $T$,
dependent on $a, c, \varepsilon$,
such that if $f(t)\neq 0$ for all $t\geq T$, then either
$L(a) < c(1-\varepsilon)/2\pi^2$ and this is witnessed by
natural numbers $n, m$ with $m < T/2\pi$, or 
$L(a)\geq c(1-\varepsilon)/2\pi^2$.
\end{lemma}

We now use the above lemmas to show the central result of
this section:

\begin{theorem}
Fix a positive real algebraic number $a$. If the 
Continuous Skolem Problem is decidable 
then $L(a)$ may be computed to within arbitrary precision.
\label{thm:hard}
\end{theorem}
\begin{proof}
Suppose we know $L(a)\in[p, q]$ for non-negative
$p,q\in\mathbb{Q}$. Choose a real algebraic $c$ with $c>0$ and a rational
$\varepsilon\in(0,1)$ such that
\[
p < \frac{c(1-\varepsilon)}{2\pi^2} < \frac{c}{2\pi^2(1-\varepsilon)}
< q.
\]
Write $A \defn c(1-\varepsilon)/2\pi^2$ and $B \defn c/2\pi^2(1-\varepsilon)$.
Calculate the maximum of the thresholds $T$ required by Lemmas
\ref{hardForward} and \ref{hardBackward}.  Check for all denominators
$m\leq T/2\pi$ whether there exists a numerator $n$ such that $n, m$
witness $L(a) \leq A$.  If so, then continue the approximation procedure
recursively with confidence interval $[p, A]$.  Otherwise, use the
oracle for the Continuous Skolem Problem to determine whether at least
one of $f_1(t), f_2(t)$ has a zero on $[T, \infty)$. If this is the
case, then $f(t)$ also has a zero on $[T, \infty)$, so by Lemma
\ref{hardForward}, $L(a)\leq B$ and we continue the
approximation recursively on the interval $[p, B]$.  If not,
then $L(a)\geq A$ by Lemma \ref{hardBackward}, so we continue on
the interval $[A, q]$. Notice that in this procedure, one can
choose $c, \varepsilon$ at each stage in such a way that the confidence
interval shrinks by at least a fixed factor,
whatever the outcome of the oracle invocations. It follows
therefore that $L(a)$ can be approximated to within arbitrary
precision.
\end{proof}

We conclude by remarking that the exponential polynomials $f_1$ and
$f_2$ involved in the proof of Theorem~\ref{thm:hard} involve only two
rationally linearly independent frequencies.  Thus the theorem applies
as soon as we have a decision procedure for exponential polynomials
with two rationally linear independent frequencies.

\section{Proof of Lemma~\ref{lem:combine}}
\label{sec:missing}
We divide Lemma~\ref{lem:combine} into two separate results, which are
proven below.
\begin{proposition}
  There is a procedure that, given a semi-algebraic set
  $D\subseteq \mathbb{R}^{n+1}$ and real algebraic numbers
  $r_1,\ldots,r_n$, returns an integer $T$ such that
  $\{ t \geq 0 : (t,e^{r_1t},\ldots,e^{r_nt}) \in D \}$ either contains
  the interval $(T,\infty)$ or is disjoint from $(T,\infty)$.  The
  procedure also decides which of these two eventualities is the case.
\label{prop:o-minimal}
\end{proposition}
\begin{proof}
  Consider a non-zero polynomial $P \in \mathbb{R}[x_0,\ldots,x_n]$ whose
  coefficients are real algebraic numbers.  Then we can write 
  $P(t,e^{r_1t},\ldots,e^{r_nt})$ in the form
  \[ Q_1(t)e^{\beta_1 t} + \ldots + Q_{m}(t) e^{\beta_{m}t} \]
  for non-zero univariate polynomials $Q_1,\ldots,Q_{m}$ 
with real algebraic coefficients and real
  algebraic numbers $\beta_1>\ldots >\beta_{m}$.  It is clear that
  for $t$ sufficiently large, $P(t,e^{r_1t},\ldots,a^{r_nt})$ has the same
  sign as the leading term $Q_1(t)$.  The proposition easily
  follows.
\end{proof}

\begin{lemma}
  Let $g : D \rightarrow \mathbb{R}$ be a bounded 
  semi-algebraic function with domain $D \subseteq \mathbb{R}^n$.  Let
  $\vr=(r_1,\ldots,r_n)$ be a tuple of real algebraic numbers and
  $T_1$ an integer such that
  $e^{\vr t} = (e^{r_1t},\ldots,e^{r_nt}) \in D$ for all $t>T_1$.
  Then the limit $g^* = \lim_{t\rightarrow \infty} g(e^{\vr t})$
  exists, is an algebraic number, and there are effective constants
  $T_2,\varepsilon>0$ such that
  $|g(e^{\vr t})- g^*| < e^{-\varepsilon t}$ for all $t>T_2$.
\label{lem:convergence}
\end{lemma}
\begin{proof}
  Since $g$ is semi-algebraic, there is a non-zero polynomial $P$
  with real algebraic coefficients such that $P(\vx,g(\vx))=0$ for
  all $\vx \in D$ (see~\cite[Proposition 2.86]{BasuPR06}).  In particular,
  we have $P(e^{\vr t},g(e^{\vr t}))=0$ for all $t>T_1$.  Gathering
  terms, we can rewrite this equation in the form
\[ Q_1(g(e^{\vr t}))e^{\beta_1 t} +
\ldots + Q_{m}(g(e^{\vr t}))e^{\beta_{m} t} = 0\] for non-zero univariate
polynomials $Q_1,\ldots,Q_{m}$ with real-algebraic coefficients
and real algebraic numbers
$\beta_1>\ldots >\beta_m$.

If $m=1$ then for all $t>T_1$ we have $Q_1(g(e^{\vr t}))=0$.  Thus
$g(e^{\vr t})$ is equal to some of root of $Q$ for all $t>T_1$.  Then
by Proposition~\ref{prop:o-minimal} there exists $T_2$ such that
$g(e^{\vr t})=g^*$ for some fixed root $g^*$ of $Q_1$ and all $t>T_2$.

If $m>1$, since $g$ is a bounded function, for all $t>T_1$ we have 
\begin{eqnarray} 
\big|Q_1(g(e^{\vr t}))\big|  &=& \left|Q_2(g(e^{\vr t})) e^{(\beta_2-\beta_1)t} + \ldots +  
Q_m(g(e^{\vr t})) e^{(\beta_m-\beta_1)t} \right| \notag \\
& \leq & M e^{(\beta_2-\beta_1)t} \label{eq:exp-bound}
\end{eqnarray}
for some constant $M$.  If $Q_1$ has degree
$d$ then (\ref{eq:exp-bound}) implies that the closest root of
$Q_1$ to $g(e^{\vr
  t})$ has distance at most
$(Me^{\beta_2-\beta_1})^{1/d}$.  Hence there exists a root
$g^*$ of $Q_1$ and effective constants
$\varepsilon,T_2>0$ such that
$|g(e^{\vr t})-g^*| < e^{-\varepsilon t}$ for all $t>T_2$.  
\end{proof}

\section{One Linearly Independent Frequency}
\label{sec:one-osc}

\begin{theorem}\label{thm:mtargument}
  Let
  $f(t)=\sum_{j=1}^n
  e^{r_jt}\left(P_{1,j}(t)\cos(\omega_jt)+P_{2,j}\sin(\omega_j)\right)$
  be an exponential polynomial such that $\omega_1,\ldots,\omega_n$
  are all integer multiples of single real algebraic number $b$.  Then we
  can decide whether or not $\{t\in \mathbb{R}_{\geq 0}:f(t)=0\}$ is
  bounded and, if bounded, we can compute an integer $T$ such that
  $\{t\in \mathbb{R}_{\geq 0}:f(t)=0\}\subseteq [0,T]$.
\end{theorem}
\begin{proof}
  Recall that for each integer $n$, both $\cos(nbt)$
  and $\sin(nbt)$ can be written as polynomials in $\sin(bt)$ and
  $\cos(bt)$ with integer coefficients.  It follows that we can write
  $f$ in the form
\[ f(t) = Q(t,e^{r_1t},\ldots,e^{r_nt},\cos(bt),\sin(bt)) \, , \]
for some multivariate polynomial $Q$ with real algebraic coefficients
that is computable from $f$.

Write $\mathbb{R}_{\geq 0}$ as the union of
$R_{+} := \{ t\in\mathbb{R}_{\geq 0}: \sin(bt)\geq 0\}$ and
$R_{-}:=\{t\in\mathbb{R}_{\geq 0} : \sin(bt)\leq 0 \}$.  We show how
to decide boundedness of $\{t\in R_{+}: f(t)=0\}$.  The analogous case
for $R_{-}$ follows \emph{mutatis mutandis}.  The idea is to partition
$\{t\in R_{+}:f(t)=0\}$ into components $Z_1,\ldots,Z_m$ and to show
for how to decide boundedness of each $Z_j$.

To this end define a semi-algebraic set 
\[
E := \left\{ (\boldsymbol{u},x) \in \mathbb{R}^{n+2} : \exists y 
\geq 0 \left(x^2+y^2=1 \wedge Q(\boldsymbol{u},x,y)=0\right) \right\} \, .
\]
Then for all $t\in R_{+}$  we have $f(t)=0$ if and only if
$(t,e^{\va t},\cos(bt))\in E$.

Consider a cell decomposition $E =C_1 \cup \ldots \cup C_m$, for cells
$C_1,\ldots,C_m \subseteq \mathbb{R}^{n+2}$, and define
\[ Z_j = \{ t\in R_{+} : \left(t,e^{\vr t},\cos(bt)\right) \in C_j \}
\]
for $j=1,\ldots,m$.  Then
$\{t\in R_{+} : f(t)=0\} = Z_1\cup\ldots\cup Z_m$.

We now focus on $Z_j$ for some fixed $j\in\{1,\ldots, m\}$ and show how
to decide whether or not $Z_j$ is finite.  Write
$D_j \subseteq \mathbb{R}^{n+1}$ for the projection of the corresponding
cell $C_j \subseteq \mathbb{R}^{n+2}$ onto the first $n+1$
coordinates.  Then $Z_j \subseteq \{ t \in \mathbb{R}_{\geq 0} :
(t,e^{\vr t}) \in D_j\}$.

First, suppose that
$\{ t \in \mathbb{R}_{\geq 0} : (t,e^{\vr t}) \in D_j \}$ is bounded.
Since $D_j$ is semi-algebraic, by Proposition~\ref{prop:o-minimal} we
can compute an upper bound $T$ of this set.  In this case $t < T$
whenever $t\in Z_j$.

On the other hand, suppose that
$\{ t \in \mathbb{R} : (t,e^{\vr t}) \in D_j \}$ is unbounded.  Then, by
Proposition~\ref{prop:o-minimal}, this set contains an unbounded
interval $(T,\infty)$.  We claim that in this case $Z_j$ must be
unbounded.  
There are two cases according to the nature of the cell
$C_j$.

\subsection*{Case I: $C_j$ is a $(\ldots,0)$-cell.}
In this case there is a continuous semi-algebraic
function $g : D_j \rightarrow \mathbb{R}$ such that
$C_j = \{ (\boldsymbol{u},g(\boldsymbol{u})) : \boldsymbol{u}\in D_j\}$.
Then for $t \in R_{+} \cap (T,\infty)$,
\begin{eqnarray*}
g(t,e^{\vr t})=\cos(bt) & \Longleftrightarrow & 
(t,e^{\vr t},\cos(bt)) \in C_j\\
& \Longleftrightarrow & f(t)=0 \, .
\end{eqnarray*}
In other words, $f$ has a zero at each point $t\in R_{+}\cap(T,\infty)$ at
which the graph of $g(t,e^{\vr t})$ intersects the graph of
$\cos(bt)$.  Since $g$ is a continuous function with an unbounded
domain that takes values in $[-1,1]$, there are infinitely
many such intersection points--see Figure~\ref{fig:intersection}.

\subsection*{Case II: $C_j$ is a $(\ldots,1)$-cell.}
In this case $C_j$ contains some $(\ldots,0)$-cell and so the argument
in Case I shows that $f$ has infinitely many zeros in $Z_j$.
\end{proof}
 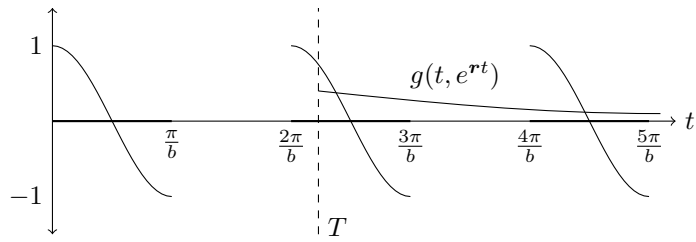
\begin{figure}
 \begin{center}
 \begin{tikzpicture}
   \draw[->] (0,0) -- (8.2,0) node[right] {$t$};
   \draw[<->] (0,-1.5) -- (0,1.5) ;
   \draw (0,1) node[left] {$1$};
   \draw (0,-1) node[left] {$-1$};

\draw (1.57,0) node[below] {$\frac{\pi}{b}$};
\draw (3.14,0) node[below] {$\frac{2\pi}{b}$};
\draw (4.71,0) node[below] {$\frac{3\pi}{b}$};
\draw (6.28,0) node[below] {$\frac{4\pi}{b}$};
\draw (7.85,0) node[below] {$\frac{5\pi}{b}$};

\draw[thick] (0,0)--(1.57,0);
\draw[thick] (3.14,0)--(4.71,0);
\draw[thick] (6.28,0)--(7.85,0);

  \draw plot[domain=0:1.57] (\x,{cos(2 * \x r)});

 \draw plot[domain=0:1.57] (\x+3.14,{cos(2 * \x r)});

 \draw plot[domain=0:1.57] (\x+6.28,{cos(2 * \x r)});

   \draw (3.5,-1.4) node[right] {$T$};
   \draw (3.5,0.4) sin (8,0.1);
 \draw[dashed] (3.5,-1.5) -- (3.5,1.5);

   \draw (5.3,0.3) node[above] {$g(t,e^{\vr t})$};
 \end{tikzpicture}
 \end{center}
 \caption{Intersection points of $g(t,e^{\vr t})$ and  $\cos(bt)$ for 
$t\in R_{+}$ (with $R_{+}$ shown in bold).}
 \label{fig:intersection}
 \end{figure}

\end{document}